\documentclass[prl, twocolumn, showpacs, superscriptaddress, amsmath,amssymb, floatfix, eqsecnum,nofootinbib]{revtex4-1}
\pdfoutput=1
\usepackage{amsmath}
\usepackage{amssymb}
\usepackage{amsthm}
\usepackage{amsfonts}
\usepackage{bbm}
\usepackage{comment}
\usepackage[normalem]{ulem}
\usepackage{graphicx}
\usepackage[dvipsnames]{xcolor}
\usepackage{color,framed}
\usepackage{hyperref}
\usepackage{times}
\usepackage{enumerate}
\usepackage{lipsum}
\usepackage{slashed}
\usepackage{url}
\usepackage{bbm}
\usepackage{chngcntr}
\counterwithout{equation}{section}

\hypersetup{
    colorlinks=true, 
    linktoc=all,     
    linkcolor=blue,  
}

\def \beq {\begin{equation}}
\def \eeq {\end{equation}}
\def \beqa {\begin{eqnarray}}
\def \eeqa {\end{eqnarray}}
\def \bseq {\begin{subequations}}
\def \eseq {\end{subequations}}

\newcommand{\<}{\langle}
\renewcommand{\>}{\rangle}

\newcommand{\ket}[1]{|#1\rangle}

\usepackage{mdframed}
\newtheorem{thm}{Theorem}

\newtheorem{lemma}{Lemma}
\newtheorem{prop}{Proposition}

\theoremstyle{definition}
\newtheorem{defin}{Definition}

\def\eea{\end{eqnarray}}

\def\Tr{ {\rm Tr} }
\def\<{\langle}
\def\>{\rangle}

\def\bZ{\mathbb{Z}}
\def\bR{\mathbb{R}}

\def\[#1\]{%
  \begin{equation}\begin{gathered}#1\end{gathered}\end{equation}%
}

\begin{document}

\title{Translation symmetry-enforced long-range entanglement in mixed states}
\author{Ryan Thorngren}
\affiliation{Mani L. Bhaumik Institute for Theoretical Physics, Department of Physics and Astronomy,
University of California, Los Angeles, CA 90095, USA}
\author{Lei Gioia}
\affiliation{Walter Burke Institute for Theoretical Physics, Caltech, Pasadena, CA, USA}
\affiliation{Department of Physics, Caltech, Pasadena, CA, USA}
\author{Carolyn Zhang}
\affiliation{Department of Physics, Harvard University, Cambridge, MA 02138, USA}

\begin{abstract}
We show by a counting argument that even though translation symmetry admits symmetric short-range entangled (SRE) eigenstates, there are not enough such SRE eigenstates to span the zero momentum sector. This means that the fixed point strong-to-weak spontaneous symmetry breaking state of translation symmetry is long-range entangled: it cannot be written as a mixture of SRE states. This is a subtle form of long-range entanglement in mixed states that cannot be detected by long-range connected correlation functions.
\end{abstract}

\maketitle

\emph{Introduction}---Recent work on open quantum systems has highlighted the rich landscape of different patterns of entanglement unique to mixed states. We say a mixed state is long-range entangled (LRE) if it is not approximately a mixture of easy-to-prepare (i.e. ``short-rangle entangled'') states~\cite{hastingsfinite,tarunseparability,tarunseparability2}\footnote{Note that under this definition, states like $\frac{1}{2}\left(|\{Z_r=+1\}\rangle\langle\{Z_r=+1\}|+|\{Z_r=-1\}\rangle\langle\{Z_r=-1\}|+\right)$ is considered SRE, even though it has classical correlations precluding an SRE purification. This definition of SRE focuses on quantum entanglement. Also note that such states are sometimes called ``separable'', even though in the quantum information literature ``separable'' is reserved for states that satisfy the stronger condition that they can be written as mixtures of \emph{product states}.}.
It is natural to ask: what are general classes of LRE mixed states? What are the mechanisms that lead to this entanglement? And how can long range entanglement be detected?

Symmetries have served as an important tool for exploring the above questions~\cite{deGroot2022symmetryprotected,lessaenforced,sala2024,lessamultipartite,lessa2025higher,salastrong,ziereis2025strongtoweaksymmetrybreakingphases,PRXQuantum.6.010347}. In particular, Ref.~\onlinecite{lessamultipartite} showed that strong \emph{anomalous} symmetries preclude a state from being multipartite separable. If a mixed state has a strong $G$ symmetry, meaning $U_g\rho\propto\rho,\,\forall g\in G$, this implies that every state $|\psi_i\rangle$ in its decomposition must be LRE. This led to the construction of mixed states that, unlike any pure state, are tripartite entangled but bipartite disentangled.

On the other hand, on-site symmetries of the form $U_g=\otimes_rU_{g,r}$, where each $U_{g,r}$ acts on a single site, clearly admit completely disentangled symmetric product states. In fact, for abelian $G$, their entire symmetry sectors can be spanned by symmetric product states: $\rho\propto P_{\pi(g)}=\sum_i|\psi_i\rangle\langle\psi_i|$ where $\{|\psi_i\rangle\}$ are product states and $\pi(g)$ labels an irreducible representation of $G$. $P_{\pi(g)}$ is also known as the maximally mixed state in a symmetry sector (MMIS), and gives the fixed point representative of a strong-to-weak spontaneous symmetry breaking (SWSSB) phase~\cite{ma2025,lessa2025swssb,ziereis2025strongtoweaksymmetrybreakingphases,lessaenforced,weinstein2025}. On the other hand, strong non-abelian finite and continuous symmetries enforce entanglement in the MMIS~\cite{lessaenforced,salastrong}.

Translation symmetry lies somewhere in between on-site symmetries and anomalous symmetries. Translation symmetry has the peculiar property that states in particular symmetry sectors, i.e. with nonzero momentum (transforming as $T|\psi\rangle = e^{ik}|\psi\rangle$ with $k\neq 0$ mod $2\pi$), must be long range entangled~\cite{leimomentum}. In the zero momentum sector, it is straightforward to write down translation invariant product states. However, unlike for on-site symmetries, it may not be possible to span the zero momentum sector with short range entangled zero momentum states.

In the context of open quantum systems, the above question can be rephrased as -- are states that SWSSB translation symmetry LRE? While the results of Ref.~\onlinecite{lessamultipartite} rule out an SRE SWSSB state for anomalous symmetries, it does not rule this out for symmetries like translation and the recently introduced non-onsiteable anomaly-free finite internal symmetries~\cite{wilbur,tuanomalies}.

In this work, we show that while translation symmetry is anomaly-free and admits symmetric product states, its zero momentum sector cannot be spanned by symmetric SRE states. Therefore, the projector onto the zero momentum sector, which demonstrates SWSSB of translation symmetry, is LRE. This long range entanglement is rather subtle because it is not detected by long range connected correlation functions~\cite{supp}. The mismatch between the space of zero momentum states and the space spanned by SRE zero momentum states gives a new mechanism for enforcing long-range entanglement. 

It is worth mentioning that from the perspective of parent channels for a mixed quantum state, which describe dynamics steering arbitrary states to the mixed state of interest, strong translation symmetry is very unusual. A quantum channel $\mathcal{E}[\rho]=\sum_iK_i \rho K_i^\dagger$ is strongly symmetric if $[U_g,K_i]=0$ for all $g,i$. This is impossible for translation symmetry unless $K_i$ has support over the entire system. Indeed, we will give an example where translation SWSSB states show up naturally as steady states of quantum channels, and these quantum channels are nonlocal.

Throughout, we work in 1d with periodic boundary conditions, on a chain of length $n$ with local Hilbert space/physical dimensions $q$. However, our results generalize to higher dimensions (see \emph{Discussion}).

\emph{Main Theorems}---We study the maximally mixed state of translation invariant states, given by
    \begin{align}
        \rho_{\rm TI}= \frac{1}{D_{\rm TI}}\sum_j\ket{\Psi_j}\langle\Psi_j| \,,
    \end{align}
where $j$ sums over an orthonormal basis of translation-invariant (i.e. zero momentum) states $\ket{\Psi_j}$ and $D_{\rm TI}$ is the dimension of the space of these states. We will prove $\rho_{\rm TI}$ is long-range entangled in the sense of the following theorem.

\begin{mdframed}[backgroundcolor=green!10]
\begin{thm}\label{thm1}
    The maximally mixed state of translation invariant states $\rho_{\rm TI}$ is not well-approximated by any mixture of easily-preparable pure states.
    
    In particular, if $\sigma_
    \tau$ is a mixture of pure states each preparable by time-$\tau$ evolution from a product state by a range-$r$ Hamiltonian (not necessarily all the same), and $\sigma_\tau$ well-approximates $\rho_{\rm TI}$ in the sense that
    \[D(\rho_{\rm TI},\sigma_\tau)\le 1 - \eta\,,\]
    where $D$ is trace distance and $\eta > 1/\text{poly}(n)$, then 
    \[\tau =\Omega(\frac{\sqrt{n}}{\text{polylog}(n)})\,,\]
    meaning $\tau$ is lower bounded by a function which is asymptotic to $\sqrt{n}/\text{polylog(n)}$.
\end{thm}
\vspace{2mm}
\end{mdframed}
The theorem means that to sample states from an ensemble with even very mild fidelity with $\rho_\text{TI}$ requires a time growing quickly with the system size. The $\sqrt{n}$ lower bound essentially comes from the volume of a light-cone for time-$\tau$ evolution being size $\tau^2$.

The first ingredient in the proof is a lemma which is generally useful when proving long-range-entanglement. It says that it is actually enough to show $\rho_\text{TI}$ is not well approximated by ``depth-$d$ states'' for large enough $d$. These are pure states created by applying a depth-$d$ circuit to a product state. This lemma based on the fact that time-$\tau$ evolution under a bounded geometrically local time-dependent Hamiltonian is well approximated by a depth-$d$ circuit for $d \le d(\tau,n,\epsilon,r)$, a universal function depending only on the number of sites $n$, the time $\tau$ of evolution, and the range $r$ of the Hamiltonian \cite{Haah_2021}. See Theorem \ref{approximationthm} in the Supplemental Material for the precise statement and the proof of the following.

\begin{mdframed}[backgroundcolor=green!10]
\begin{lemma}\label{lemmatails} 
\textup{\bf{(Tails Lemma)}}
There is a function $d(\tau,n,\epsilon,r)=O(\tau\ \text{polylog}(n\tau/\epsilon))$ such that if $P_d$ is the projection onto the space of depth-$d$ states, $d \ge d(\tau,n,\epsilon,r)$, ${\rm Tr }P_d \rho \le \delta$, and $\sigma_\tau $ is any mixture of time-$\tau$ states (as in Theorem \ref{thm1}), then
\[D(\rho,\sigma_\tau) \ge 1-\epsilon^2-\delta\,.\]
\end{lemma}
\end{mdframed}
Intuitively, this means that if there is a small probability of measuring $P_d$ on $\rho$ and obtaining 1 (meaning that the resulting state is in the span of depth-$d$ states) then $\rho$ is far away from any mixture of time-$\tau$ states.

The second ingredient, and the main technical result of the paper, is an upper bound on the dimension of the span of all translation-invariant depth-$d$ states. This upper bound is based on the idea that a depth-$d$ circuit has $\approx n^d$ parameters, while the dimension of all translation-invariant states is exponential in $n$. However, the argument is not quite trivial, since the depth-$d$ states are a non-linear function of the circuit parameters. Indeed, one can imagine a 1-dimensional curve whose values span all of any finite dimensional vector space, such as the curve $f(t) = (1,t,t^2,\ldots, t^N)$. We thus need to have control over how the coefficient functions can vary as a function of the parameters.

As a warm-up, we consider translation-invariant matrix product states (TIMPS).

\begin{mdframed}[backgroundcolor=green!10]
\begin{prop}\label{propMPScounting}\textup{\bf{(MPS Counting)}}
    The dimension $D_{\rm TIMPS}(n,q,d)$ of the span of all TIMPS on a ring of size $n$, with bond dimension $d$, and physical dimension $q$, is bounded above by $D_{hpoly}(n,qd^2)$, which is the dimension of the space of degree $n$ homogeneous polynomials in the $qd^2$ variables. In particular,
    \begin{align}
        \begin{split}
           D_{\rm MPS}(n,q,d) \le D_{hpoly}(n,qd^2)={qd^2-1+n \choose n} \,.
        \end{split}
    \end{align}
\end{prop}
\end{mdframed}
\begin{proof}
    We write the MPS tensor $A$ in terms of $d \times d$ matrices $A=\{A^i\}_{i=1}^q$ where $q$ is the physical dimension. The (unnormalized) MPS state on $n$ $q$bits is
    \[\ket{A}=\sum_{i_1, \ldots, i_n} ({\rm Tr} A^{i_1} \cdots A^{i_n}) \ket{i_1,\ldots, i_n}\,,\]
    where the $A$ at each site is the same, guaranteeing translation symmetry.
    We observe that the coefficients of this vector are degree $n$ homogeneous polynomials in the $q d^2$ coefficients $x_1,\cdots x_{qd^2}$ of the tensor $A$ : they are linear combinations of monomials
    \begin{equation}
        x_1^{\alpha_1}x_2^{\alpha_2}\cdots x_{qd^2}^{\alpha_{qd^2}}
    \end{equation}
    with $\alpha_1+\alpha_2+\cdots\alpha_{qd^2}=n$. The number of these monomials is
    \begin{align}
        \begin{split}
            D_{hpoly}(n,qd^2)&={qd^2+n-1 \choose n} \\
            &= \frac{1}{n!} (qd^2+n-1) \cdots (qd^2).\label{Dhpoly}
        \end{split}
    \end{align}
    There are $q^n$ components in the vector $\ket{A}$, so if $q^n > D_{hpoly}(n,qd^2)$, then these component coefficients satisfy
    \[q^n - D_{hpoly}(n,qd^2)\,,\]
    linear relations \emph{as functions of the coefficients of A}. Therefore, if we plug in any values for these coefficients, obtaining an MPS state, these states all satisfy the same $q^n - D_{hpoly}(n,qd^2)$ linear relations. Thus, the proposition follows.
\end{proof}

If we could apply Proposition \ref{propMPScounting} to translation-invariant depth-$d$ states we would be done. Unfortunately, we do not know if such states are always expressible as TIMPS. Indeed, we do not assume that the depth-$d$ circuit itself has any sort of translation-invariance---only that the resulting state does. Moreover, a $\log n$ depth circuit can give rise to logarithmic entanglement entropy, indicating a bond dimension $d=\mathcal{O}(n)$. This would appear sufficient according to Proposition \ref{propMPScounting} to span the space of translation invariant states. Using a cutting argument, we will show below that translation invariance and the circuit structure allow for a formulation of the resulting MPS as a translation-invariant MPS with much smaller bond dimension than the naive counting above.

\begin{mdframed}[backgroundcolor=green!10]
\begin{prop}\label{propSREbound}\textup{\bf{(Circuit Counting)}}
    The dimension $D_{\rm SRE}(n,d,q)$ of the span of translation-invariant depth-$d$ SRE states on a ring of length $n = m(2d+1) + r$, where $1 \le r \le 2d+1$, with local dimension $q$, is bounded above by
    \begin{align}
        \begin{split}
            D_{\rm SRE}(n,d,q)& \le 2d(2d+1) q^4 D_{hpoly}(m,2d(2d+1)q^4)\\
            &= 2d(2d+1) q^4 {2d(2d+1)q^4 -1 + m \choose m}.
        \end{split}
    \end{align}
\end{prop}
\end{mdframed}

\begin{proof}
    Let $\ket{\psi}$ be a depth-$d$ state obtained by
    \[\ket{\psi}=C \ket{0}\,,\]
    where $C$ is a depth-$d$ circuit and $\ket{0}$ is a translation-invariant product state. We assume that $\ket{\psi}$ is invariant under the shift operator $T$ but make no other assumption on $C$.
    
    Consider a particular site labeled $0 \in \bZ_n$. We consider the sub-circuit $C_0$ defined by keeping all of the gates of $C$ in the intersection of the light-cones of site 0 and 1. Here is an example, for depth $d = 4$, where the circuit elements of $C_0$ are drawn in blue and the rest of the circuit elements of $C$ are drawn in green:
    \begin{center}
        \includegraphics[width=9cm]{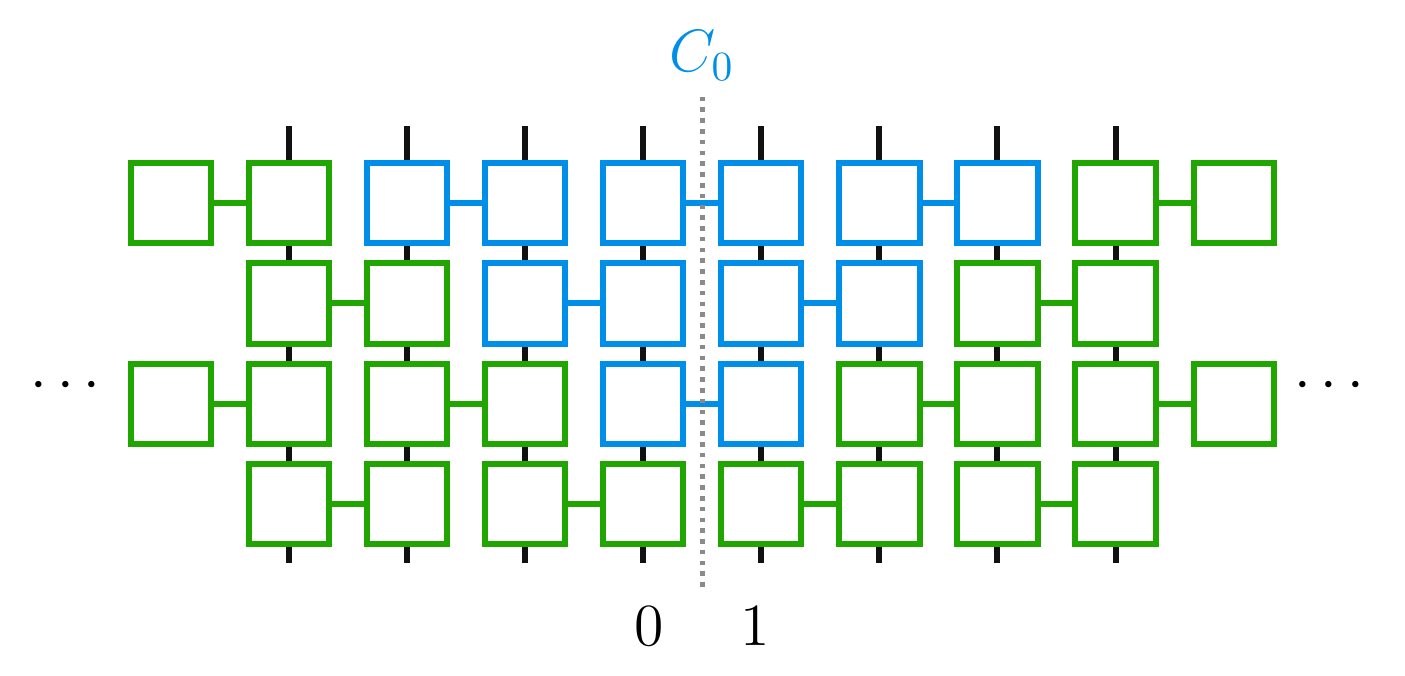}
    \end{center}
    Note, we have drawn the circuit elements with an internal bond representing their Schmidt decompositions. Thinking of these circuits as composed of matrix-product operators (MPOs) will be useful below.
    
    Observe that we can remove $C_0$ from $C$ by considering $|\psi_0\rangle =C_0^{-1}\ket{\psi}$
    which separates the sites 0 and 1. The result is a state which is ``cut'', meaning that if we consider the reduced density matrix $\rho_I$ to an interval $I$ containing 0 and 1, whose complement is at least length $2d+1$, then $\rho_I$ separates into a tensor product of a left interval containing $0$ and a right interval containing $1$. This can be seen by inspecting the circuit presentation of $|\psi_0\rangle$.

    Because $\ket{\psi}$ is translation invariant, the same property holds across a cut at $x+1/2$ for the state $T^x C_0^{-1} T^{-x} \ket{\psi}$
    where $T$ is the translation operator taking site $x$ to site $x+1$. In fact, let
    $C_x := T^x C_0 T^{-x}$ and consider $C_x^{-1} C_y^{-1} \ket{\psi}$ with $|x-y| \ge 2d+1$ (assume $y > x$). Then by the argument above, this state satisfies the cut property at $x+1/2$ and at $y + 1/2$, since the two cutting operators $C_x$ and $C_y$ are not overlapping. As a result, the state splits as a tensor product of a state on the interval $[x+1,y]$ and its complement.
    
    Now we apply this cutting many times. We write
    \[n = (2d+1)m+r\,,\]
    where $1 \le r \le 2d+1$ (this funny convention is to avoid casework below), and consider the state
    \[\label{eqncutstate}\ket{\psi_{\rm cut}}=\prod_{j=0}^{m-1} C_{(2d+1)j}^{-1} \ket{\psi}\,.\]
    By the same argument above,
    \[\label{eqnfactorization}\ket{\psi_{\rm cut}} = \left( \bigotimes_{j=0}^{m-1} T^{(2d+1)j}\ket{\psi_{\rm block}} \right) \otimes \ket{\psi_{\rm remainder}}\,,\]
    where $\ket{\psi_{\rm block}}$ is a state on $[1,2d+1]$ and $\ket{\psi_{\rm remainder}}$ is a state on $[(2d+1)m+1,n]$. We will call these intervals the ``blocks''. Note that when the size of the remainder happens to be $r = 2d+1$ (so $n$ is divisible by $2d+1$), this is a translation-invariant state with a unit cell of size $2d+1$. Otherwise, it is approximately translation-invariant, in that $T^{2d+1}$ takes one block state to the next, except at the end where the remainder is. This follows from the expression \eqref{eqncutstate} and the translation-invariance of $\ket{\psi}$.
    \begin{center}
        \includegraphics[width=6cm]{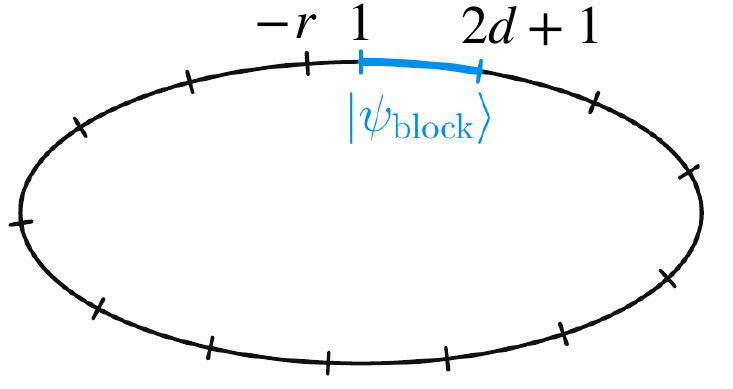}
    \end{center}
    
    Inverting \eqref{eqncutstate} we get
    \[\ket{\psi} = \prod_{j=0}^{m-1} C_{(2d+1)j} \ket{\psi_{\rm cut}}\,.\]
    Thinking of the blocks as large physical sites, we can express this circuit contraction as a matrix product state
    \[\ket{\psi} = \sum_{i_1,\ldots,i_m,j} (\Tr A^{i_1} \cdots A^{i_m} B^j) \ket{i_1,\ldots,i_m,j}\,,\]
    where $i_l$ runs over the basis vectors of the $l$th block $[l+1,l+2d+1]$ and $j$ over the basis of the remainder block. We can already apply the argument of Proposition \ref{propMPScounting} to these states, although this will be very inefficient.
    In particular, the problem arises because although we have broken up the chain into length $2d+1$ blocks in $|\psi_{\mathrm{cut}}\rangle$, each $A^{i_j}$ is a very large matrix with dimension exponential in $d$. This would indicate only a $\text{polylog}(n)$ lower bound on the depth $d$ needed to span the space of translation invariant states in the zero momentum sector. In the following, we will show that because $|\psi_{\mathrm{block}}\rangle$ is not an arbitrary state on $2d+1$ sites, we would actually need $d\sim\sqrt{n}$.

    To do better, note that $\ket{\psi_{\rm block}}$ and $\ket{\psi_{\rm remainder}}$ can each be obtained by applying a depth $2d$ circuit to the product state $\ket{0}$ on an interval containing the block. The depth is $2d$ because we applied the depth-$d$ circuit in \eqref{eqncutstate} to the state $\ket{\psi}$ which already had depth-$d$.
 
    Given such a state, we obtain a Schmidt decomposition by decomposing each circuit element into a two-body MPO with a bond dimension $q^2$:
    \begin{center}
        \includegraphics[width=4cm]{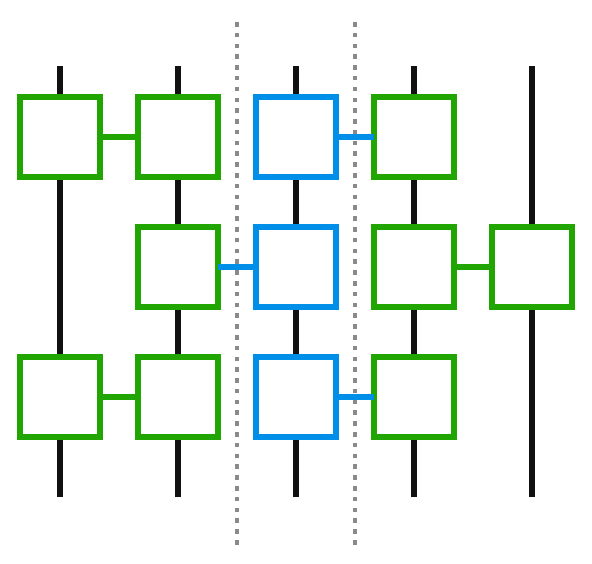}
    \end{center}
    In this decomposition, there are $2d(2d+1)$ of these tensors for $A$ and $2dr$ of these tensors for $B$, each having a total of $q^4$ components. Thus, we can think of $A$ as being composed of the
    \[2d(2d+1)q^4\]
    components of these tensors, and $B$ composed of
    \[2dr q^4\]
    components. As in Proposition \ref{propMPScounting}, the dimension of the span of states $\ket{\psi}$ is given by degree $m+1$ polynomials with $m$ degrees from the $2d(2d+1)q^4$ variables of $A$ and 1 degree from $r 2d q^4$, giving us an upper bound
    \begin{align}
        \begin{split}
            D_{\rm SRE}(n,d,q) &\le r 2d q^4 D_{hpoly}(m,2d(2d+1)q^4) \\
            &= r 2d q^4 {2k(2k+1)q^4 + m-1 \choose 2d(2d+1)q^4-1} \\
            &\le 2d(2d+1)q^4 {2d(2d+1)q^4 + m-1 \choose 2d(2d+1)q^4-1}\,.
        \end{split}
    \end{align}
    In the last line we have removed the dependence on the remainder $r$ using $r \le 2d+1$.
\end{proof}

This upper bound is quite strong. Indeed, the binomial coefficient is a degree $\approx d^2$ polynomial in $n$. If we compare with the dimension $D_{\rm TI}(n,q)$ of all translation-invariant states, $D_{\rm TI}(n,q) \ge \frac1n q^n$ grows exponentially in $n$ (this latter dimension can be computed explicitly as the number of $q$-ary length $n$ necklaces, see the Supplemental Material~\cite{supp}). Intuitively, this means that the projection of $\rho_{\rm TI}$ onto the space of depth-$d$ SRE states will be exponentially small unless $d \approx \sqrt{n}$, leading to the proof of Theorem~\ref{thm1}.

\textit{Proof of the Main Theorem---}By Lemma \ref{lemmatails}, for all $\epsilon > 0$, $d \ge d(\tau,n,\epsilon)$,
    \[D(\rho_\text{TI},\sigma_\tau) \ge 1- \epsilon^2-\text{Tr }P_d \rho_\text{TI}\,.\]
    Combining with the assumption, we get
    \[\text{Tr }P_d \rho_\text{TI} \ge \eta - \epsilon^2\,.\]
    For concreteness specify to $\epsilon = \sqrt{\eta}$ to obtain
    \[\label{eqnaccuracylowerbound}\text{Tr }P_d \rho_\text{TI} \ge \frac{\eta}{2}\,.\]
    The strongest upper bound on the left-hand-side comes from $d = d(\tau,n,\sqrt{\eta})$, so specify to this value as well.

    Now we can use our dimension bounds. Since $\rho_{\rm TI}$ is an equal mixture of all translation-invariant states, we have
    \[{\rm Tr} P_d \rho_{\rm TI} = D_{\rm SRE}(n,d,q)/D_{\rm TI}(n,q).\]
    Using $D_{\rm TI}(n,q) \ge \frac1n q^n$ and Proposition \ref{propSREbound}, we find
    \begin{align}
        \begin{split}
            \text{Tr }P_d \rho_\text{TI} &< 2d(2d+1) n q^{-n+3} e^{2d(2d+1)q^4-1} \\
            &\times\left(1 + \frac{n}{(2d(2d+1)q^4-1)^\gamma}\right)^{2d(2d+1)q^4-1},
        \end{split}
    \end{align}
    for some $q$-dependent constant $1<\gamma<2$ (see Supplemental Material \cite{supp} for more details). Set $a = 2d(2d+1)q^4-1$. Then, together with \eqref{eqnaccuracylowerbound} we have
    \[\label{eqnimplicitlowerbound}\eta q^n < 2 a n e^a (1+n/a^\gamma)^a.\]
    This gives a (rather implicit) lower bound on $d=d(\tau,n,\sqrt{\eta})$, hence a (even more implicit) lower bound on $\tau$. Analyzing it asymptotically with the estimate in Lemma \ref{lemmatails}, we find when $\eta > 1/\text{poly}(n)$,
    \[\tau > \frac{\sqrt{n}}{\text{polylog}(n)}+ \cdots\]
    where $\cdots$ are asymptotically smaller terms. See the Supplemental Material. Unfortunately, to have a better error estimate of the lower bound we would need to better understand the estimate in Lemma \ref{lemmatails}.\qed

\emph{Translation SWSSB via thermalization---}We will now review a natural context in which SWSSB of translation symmetry occurs: thermalizing Floquet dynamics. The spectral form factor (SFF) is a useful diagnostic for probing thermalization in closed quantum systems. For a Floquet system described by a Floquet unitary $W$, the SFF at time $t$ is defined as
\begin{equation}
    K(t)=\overline{|\mathrm{Tr}(W^t)|^2}=\mathrm{tr}(\mathbb{T}^L)\,,
\end{equation}
where $\overline{\cdot}$ indicates a suitable average over similar circuits and $\mathbb{T}$ is a transfer matrix giving evolution in space~\cite{garratt2021,garratt2021feynman}. $\mathbb{T}$ acts in a naturally doubled Hilbert space coming from including both the forward and backward time evolutions $W$ and $W^\dagger$.

For dual-unitary Floquet circuits, $K(t)$ was shown to demonstrate precisely the linear ramp expected from random matrix theory~\cite{bertini2018,bertini2021}. In Ref.~\onlinecite{bertini2021}, it was also made clear that for dual unitary circuits, $\mathbb{T}$ is a nonlocal quantum channel represented on a doubled Hilbert space as an operator: $\mathbb{T}=\mathbb{U}_2\cdot\mathbb{O}_1^\dagger\cdot \mathbb{W}_2\cdot\mathbb{O}_0$, where $\mathbb{U}_2, \mathbb{W}_2$ are unitaries coming from spacetime rotation of the dual-unitary gates, and $\mathbb{O}_0,\mathbb{O}_1$ take the form of vectorized unital, trace-preserving nonlocal channels coming from the averaging. Folding $\mathbb{T}$ back into a channel, \cite{bertini2021} showed that its steady states are spanned by the temporal shift operators~\cite{bertini2018,bertini2021}, giving density matrices corresponding to projectors onto different eigenspaces of the temporal shift operator. This is aligned with the characterization in Refs.~\onlinecite{garratt2021,garratt2021feynman} of thermalization as SSB of $\mathbb{Z}_t\times\mathbb{Z}_t$ down to a diagonal $\mathbb{Z}_t$, or SWSSB of $\mathbb{Z}_t$ time translation symmetry. 

\emph{Discussion}---While we focused on 1d spin chains, it is not hard to see that translation symmetry in higher dimensions also does not admit SRE MMIS. In particular, if we consider a 2d system with translation symmetry in, say, the $y$-direction, then it is equivalent to a 1d spin chain with very large $q\sim\mathcal{O}(\exp(L))$. However, treating $L_y$ as $n$, the same proofs above go through to show that the space of translation invariant SRE states grows too slowly with $L_y$ to span the space of translation invariant states.

The states we discuss in this paper are very specific, but similar results apply to more general states such as those that have SWSSB for translation symmetry together with another discrete symmetry, such as on-site $\mathbb{Z}_2$ (i.e. $\rho\propto \rho_{TI}\cdot \frac{1+\prod_rX_r}{2}$). Our results may also have implications for states with strong non-invertible symmetries that mix with translation symmetry, such as Kramers-Wannier duality~\cite{seiberg2024,inamura2026}. 

It would be interesting to directly relate long range entanglement of the MMIS to obstructions to on-siteability of the symmetry, for finite abelian symmetries~\cite{wilbur,tuanomalies}. By on-siteable, we mean that there exists an on-site representation on ancillas $U_{g,a}$ and a locality preserving unitary on the composite spin-ancilla system $V$ such that $V(U_g\otimes U_{g,a})V^\dagger$ is on-site for all $g\in G$. Clearly, if $U_g$ is on-siteable without the ancillas, then it follows directly from the SRE nature of the MMIS of the on-site symmetry that the MMIS of $U_g$ is also SRE. On the other hand, one might be able to prove, perhaps using a similar counting argument as for translations, that the MMIS of the non-onsiteable 2d $\mathbb{Z}_N$ symmetries of Ref.~\onlinecite{wilbur,tuanomalies} must be LRE. The counting argument would be more subtle because for the fixed point non-onsiteable $\mathbb{Z}_2$ symmetry in Ref.~\onlinecite{wilbur}, it is easy to show that there is an exponential (rather than polynomial) number of SRE states in the even parity sector. One would have to show that this exponential is not fast enough to span the even parity sector. More generally, one might conjecture that a finite abelian symmetry has an SRE MMIS if and only if it is on-siteable.

Another reasonable conjecture related to the result in this paper is that symmetries with infinite range have LRE MMIS. For example, the operator $T\otimes T^{-1}$ implements no net translation, yet also does not admit an SRE MMIS. It would also be interesting to try to prove that $F\cdot T$ for a finite depth circuit $F$ satisfying $(F\cdot T)^n=\mathbf{1}$ also does not admit an SRE MMIS.

A curious property of translation symmetry, as mentioned in the introduction, is that its nonzero momentum sectors do not admit any SRE eigenstates~\cite{leimomentum}. It follows that the projectors onto the nonzero momentum sectors are clearly LRE. It would be interesting to see if this observation can lead to an alternative proof that the projector onto the zero momentum sector is also LRE.

\emph{Acknowledgments}---We thank Jacob Fridolin Steiner and Ruben Verresen for helpful discussions. R.T. acknowledges support from the Mani L. Bhaumik Presidential Term Chair in Theoretical Physics at UCLA. L.G. acknowledges support from the Walter Burke Institute for Theoretical Physics at Caltech and the Caltech Institute for Quantum Information and Matter. C.Z. is supported by the Harvard Society of Fellows. 

\emph{Note added}---While finalizing this manuscript, we became aware of a related work Ref.~\cite{lessa2026mixedstatelongrangeentanglementdimensional}, which will appear in the same arXiv posting.

\bibliography{main}

\onecolumngrid

\section{Supplemental Material}

\subsection*{Hamiltonian Evolution}

The following definitions are adapted from \cite{Haah_2021}, whose techniques we use.
\begin{defin}
    A \textbf{$D$-dimensional lattice $\Lambda$ of $n$ qubits} is a collection of $n$ points in $\bR^D$ all at least separated by a distance 1. In particular, every unit ball has at most one qubit.
\end{defin}

\begin{defin}
    A \textbf{range-$r$, time-$T$ Hamiltonian on $\Lambda$} is an assignment for every subset $X \subset \Lambda$ a Hermitian operator $h_X(t)$, $t \in [0,T]$ acting on the qubits in $X$, with $||h_X(t)|| \le 1$ (in operator norm) for all $X$ and $t$, and $h_X(t) = 0$ if $diam(X) > r$. We also require the Hamiltonian to be piecewise-slowly-varying over $[0,T]$, meaning over a finite set of open intervals $0 < t_1 < \cdots < t_m < 1$, $\frac{d}{dt} h_X(t)$ exists and is bounded above in norm by $1/(t_{k+1}-t_k)$.
\end{defin}

\begin{defin}
    A \textbf{range-$r$, time-$T$ quantum state} is a state on a lattice $\Lambda$ of $n$ qubits obtained by applying Hamiltonian evolution by a $D$-dimensional range-$r$ Hamiltonian over time $[0,T]$ to a product state. In particular, if $T,r\leq \mathcal{O}(\log L)$, then the state is a \textit{short-range entangled} (SRE) state.
\end{defin}

\begin{thm}\label{approximationthm}
(Theorem 1 of Ref.~\onlinecite{Haah_2021}) There is a function $d(T,n,\epsilon,r)$ such that for all range-$r$ time-$T$ Hamiltonians $H$, there is a circuit $C$ of depth at most $d(T,n,\epsilon,r)$ that approximates time evolution $U$ with error $\epsilon$, satisfying
\[\sup_\rho D(U \rho U^\dagger,C \rho C^\dagger ) < \epsilon,\]
where $D(A,B)=\frac12 ||A-B||_1 = \frac12\text{Tr }\sqrt{(A-B)^\dagger (A-B) }$ is the trace distance.\footnote{In fact, \cite{Haah_2021} proves a stronger result including ancillas. Here we set the ancilla Hilbert space to be one dimensional.} In particular, for pure states $\ket{\psi}$,
\[\label{eqnpurestateapprox}|\langle \psi| C^\dagger U \ket{\psi}|^2 > 1- \epsilon^2.\]
Moreover, $d(T,n,\epsilon,r)$ is asymptotically upper bounded as
\[d(T,n,\epsilon,r) = O(T \text{polylog}(nT/\epsilon))\]
\end{thm}
\begin{proof}
    The theorem is stated given a particular Hamiltonian in \cite{Haah_2021}. For our purposes, we would like to have a uniform upper bound on the depth needed to approximate to within $\epsilon$ by a function $d(T,n,\epsilon,r)$ that applies to all time-$T$ range-$r$ Hamiltonians simultaneously. In fact, the proof of Theorem 1 of \cite{Haah_2021} implies this, since the construction of the approximating circuit and the error upper bound above depends only on the Lieb-Robinson bound for the given Hamiltonian. Restricting ourselves to the class of time-$T$ range-$r$ Hamiltonians, there are universal Lieb-Robinson bounds, depending only on the connectivity and the maximum strength of the terms. See eg. Theorem 2.1 of \cite{Nachtergaele_2006}.
\end{proof}

\begin{lemma}\label{lemmatailsapp}
Let $P_d$ be the projection onto the space of depth-$d$ states, where $d \ge d(T,n,\epsilon,r)$. If
\[\text{Tr }P_d \rho \le \delta,\]
and $\sigma$ is any mixture of time-$T$ SRE states, then
\[D(\rho,\sigma) \ge 1-\epsilon^2-\delta\]
where
\[D(\rho,\sigma) = \frac12 ||\rho -\sigma||_{tr}\]
is the trace distance.
\end{lemma}

\begin{proof}
By Theorem \eqref{approximationthm}, every pure state $\sigma_i$ in $\sigma$ satisfies
\[\text{Tr }P_d \sigma_i \ge 1-\epsilon^2,\]
so also by linearity
\[\text{Tr }P_d \sigma \ge 1-\epsilon^2.\]
On the other hand, if $\rho$ is a state such that
\[\text{Tr }P_d \rho \le \delta,\]
then considering projective measurement of $P_d$, Theorem 9.1 of \cite{nielsen2010quantum} gives
\[D(\rho,\sigma) \ge 1-\epsilon^2-\delta\]    
\end{proof}

Note that in this lemma, we can decrease $\epsilon$ at the cost of increasing the cutoff depth $d(T,n,\epsilon,r)$. This typically causes an increase in $\delta$. Thus, there is in general some optimization to be done to get the best possible lower bound on the trace distance.

\subsection{More details and asymptotics of lower bound in Theorem \ref{thm1}}\label{appasymptoticanalysis}

First we give a derivation of the bound
\begin{align}\label{eqncomplicatedbound}
        \begin{split}
            \text{Tr }P_d \rho_\text{TI} &< 2d(2d+1) n q^{-n+3} e^{2d(2d+1)q^4-1} \\
            &\times\left(1 + \frac{n}{(2d(2d+1)q^4-1)^\gamma}\right)^{2d(2d+1)q^4-1}
        \end{split}
\end{align}
As noted already,
\[{\rm Tr} P_d\rho_{\rm TI} = D_{\rm SRE}(n,d,q)/D_{\rm TI}(n,q).\]
$D_{\rm TI}(n,q)$ is equal to the number of orbits of length-$n$ strings in the alphabet $\{1,\ldots,q\}$ under the cyclic permutation action of $\mathbb{Z}_n$. These orbits are known as ``$q$-ary necklaces of length $n$'' and a standard formula yields
\[D_{\rm TI}(n,q) = \frac1n \sum_{k|n} \varphi(k) q^{n/k}\]
where $\varphi(k)$ is the Euler totient function and the sum is over divisors of $n$. In particular, $\varphi(k) \ge 0$ and $\varphi(1)=1$ so from the $k=1$ term we get
\[D_{\rm TI}(n,q) \ge \frac1n q^n.\]

Now, from Proposition \ref{propSREbound} we have
\[D_{\rm SRE}(n,d,q) \le 2d(2d+1) q^4 {2d(2d+1)q^4 -1 + m \choose m} \\
= 2d(2d+1) q^4 {2d(2d+1)q^4 -1 + m \choose 2d(2d+1)q^4 -1}\]
where we used ${a \choose b} = {a \choose a-b}$. A standard estimate for the binomial coefficient is
\[{a \choose b} \le \left(\frac{ea}{b}\right)^b.\]
This gives
\[D_{\rm SRE}(n,d,q) \le 2d(2d+1) q^4 \left(\frac{2d(2d+1)q^4-1+m}{2d(2d+1)q^4-1} \right)^{2d(2d+1)q^4-1} \\
= 2d(2d+1) q^4 \left(1+\frac{m}{2d(2d+1)q^4-1} \right)^{2d(2d+1)q^4-1}.\]
Observe that this is an increasing function of $m$ and $m \le n/(2d+1)$. Thus
\[D_{\rm SRE}(n,d,q) \le 2d(2d+1) q^4 \left(1+\frac{n}{(2d+1)(2d(2d+1)q^4-1)} \right)^{2d(2d+1)q^4-1}.\]
Finally,
\[(2d+1)(2d(2d+1)q^4-1) \ge (2d(2d+1)q^4-1)^\gamma\]
for some $q$-dependent $\gamma \in (1,2)$. This last inequality is just to simplify some of the expressions. The bound \eqref{eqncomplicatedbound} follows.

We now analyze the $n \to \infty$ asymptotics of this lower bound \eqref{eqncomplicatedbound}. Taking logarithms,
    \[\log \eta + n \log q < \log 2 + \log a + \log n + a(1 +  \log(1+n/a^\gamma))\]
    We will assume that $a = O(n)$, since otherwise we can prove an even stronger bound than we will show. This yields
    \[\log \eta + n \log q < a(1+\log(1+n/a^\gamma)) + O(\log(n)) \\
    < a \log(1+n) + O(\log(n))\]
    Let us suppose
    \[\eta > 1/\text{poly}(n)\]
    Then we have
    \[a >\log q \frac{n}{\log n} + O(1)\]
    Recalling the definition of $a$, we may obtain
    \[(2d+1)^2 >\frac{\log q}{q^4} \frac{n}{\log n}(1  + O(\frac{\log n}{n}))\]
    so
    \[d > \sqrt{\frac{ \log q}{4q^4}}\sqrt{ \frac{n}{\log n}} (1  + O(\frac{\log n}{n}))\]
    On the other hand, by Theorem \ref{approximationthm}, 
    \[d = O(T\text{polylog}(nT/\sqrt{\eta}))\]
    To simplify further, we may assume $T = O(\sqrt{n})$ (otherwise, we have an even stronger lower bound than what we will prove) so $d = O(T \text{polylog}(n))$. Then we have
    \[T > \frac{\sqrt{n}}{\text{polylog}(n)}+ \cdots\]
    where $\cdots$ are asymptotically smaller terms. Unfortunately to have a better error estimate of the lower bound we would need to better understand the estimate in Theorem \ref{approximationthm}.

    \subsection*{Correlation functions}\label{seccorrelations}

Here we will show that $\rho_{\mathrm{TI}}$ has exponentially decaying connected correlation functions $\langle O_iO_j\rangle-\langle O_i\rangle\langle O_j\rangle$ for any local operators $O_i,O_j$. In the case of a global on-site symmetry, the expectation value of any local operator is exactly zero. For translation symmetry, local operators can give nonzero expectation values, but they are exponentially small in $n$. Specifically, we will show that 
\begin{equation}
    \frac{\mathrm{Tr}(OT^r)}{\mathrm{Tr}(\rho_k)}\leq\mathcal{O}(e^{-\mu n})
\end{equation}
for all local operators $O$ and all $r\in[1,n-1]$, for some $\mathcal{O}(1)$ constant $\mu$. Here, $\rho_k$ is the projector onto the momentum $k$ sector. This is true because the above quantity is roughly  $q^{\mathrm{gcd}(n,k)}/q^n $ which is exponentially small in $n$. In particular it is maximized for $k=n/2$, in which case it is $q^{n/2}/q^n=q^{-n/2}$ which is still exponentially small in $n$.

To be more precise, we have that
\begin{equation}
    \mathrm{Tr}(OT^r)=\sum_{\{Z_i\}}\langle \{Z_i\}|O|\{Z_{i+r}\}\rangle 
\end{equation}
Suppose that $O$ is supported on region $A$. Then the product state $|\{Z_i\}\rangle$ can only give a nonzero value above if $Z_i=Z_{i+r}$ for all $i\in\bar{A}$. This tightly constrains the number of possible nonzero values. To proceed further we divide the spin chain into $\mathrm{gcd}(n,r)$ disjoint cycles, and for each cycle $c$ we let $a_c=|A\cap c|$. Then the number of independent qudits left after imposing the constraint on $\bar{A}$ is
\begin{equation}
    d=\sum_{c=1}^{\mathrm{gcd}(n,r)}\max(1,a_c)
\end{equation}
so
\begin{equation}
    \frac{1}{q^n}|\mathrm{Tr}(OT^r)|\leq q^{d-n}
\end{equation}
which is exponentially small for any $O$ supported on $|A|$ that is $\mathcal{O}(1)$. If we allow for nonlocal operators i.e. $|A|=n-1$, then we can get $\mathcal{O}(1)$ contributions (i.e. if $O$ is an operator that shifts all qudits backward except for one, bringing $2\to n, 3\to 2, 4\to 3,\cdots n\to n-1$ we get $\frac{1}{q^n}\mathrm{Tr}(OT)=\frac{1}{q}$). 

Since all of the linear correlation functions are exponentially small, connected correlation functions are also exponentially small. This in fact holds for all $\rho_k$, so not only is the projector onto the zero momentum sector LRE in a way that is undetectable by local correlation functions, but the projectors onto the other momentum sectors are also LRE in a way undetectable by local correlation functions.

\end{document}